\documentclass[12pt]{article}

\usepackage{amstext,amsmath,amssymb,amsfonts,bbm}
\usepackage[latin1]{inputenc}
\usepackage{epsfig}
\usepackage{dsfont}
\usepackage{hyperref}
\usepackage{amsthm}
\usepackage{subfigure}
\usepackage{color}
\usepackage{multirow}
\usepackage{psfrag}
\usepackage{graphicx}
\usepackage{extarrows}
\usepackage{braket}

\usepackage[lmargin=40pt,rmargin=40pt,tmargin=80pt,bmargin=80pt]{geometry}

\theoremstyle{plain}

\newtheorem{theorem}{Theorem}

\newtheorem{proposition}{Proposition}
\theoremstyle{definition}

\allowdisplaybreaks[4]

\newcommand{\cC}{{\cal C}}
\newcommand{\cM}{{\cal M}}
\newcommand{\cG}{{\cal G}}

\begin{document}

\title{The complete $1/N$ expansion of a SYK--like tensor model}

\author{Razvan Gurau\footnote{rgurau@cpht.polytechnique.fr, 
 Centre de Physique Th\'eorique, \'Ecole Polytechnique, CNRS, Universit\'e Paris-Saclay, F-91128 Palaiseau, France
 and Perimeter Institute for Theoretical Physics, 31 Caroline St. N, N2L 2Y5, Waterloo, ON, Canada.}}

\maketitle

\abstract{A SYK--like model close to the colored tensor models has recently been proposed \cite{Witten:2016iux}. Building on results obtained in tensor models
\cite{GurSch}, we discuss the complete $1/N$ expansion of the model. We detail the two and four point functions at leading order.  
The leading order two point function is a sum over melonic graphs, and the leading order relevant four point functions are sums over dressed ladder 
diagrams. We then show that any order in the $1/N$ series of the two point function can be written solely in term of the leading order two and four point functions.
The full $1/N$ expansion of arbitrary correlations can be obtained by similar methods.}

\section{Introduction}

A tensor model version \cite{Witten:2016iux} of the SYK model \cite{Sachdev:1992fk,Kitaev} (see \cite{Polchinski:2016xgd,Jevicki:2016bwu,Maldacena:2016hyu} for recent developments) has been recently proposed by Witten.
The model is constructed starting from $D+1$ real fermionic fields $\psi^0,\dots \psi^D$, each with $n^D$ real components (yielding a total of $N=(D+1)n^D$ independent real fermionic fields).
In the tensor model literature \cite{review} the upper index $c\in \{0,\dots, D\} =  \cC $ is called the \emph{color} of the field $\psi^c$. 
For every $c$, the field $\psi^c$ lives in a vector representation of $O(n)^D$ \emph{i.e.} it is a tensor with 
$D$ indices. We denote these indices ${\bf a}^c = \left( a^{cc_1} | c_1 \in \cC \setminus \{ c \} \right)$ and each $a^{c_1c_2}$ takes values $1$ to $n$.
In full index notation, the action of the SYK--like tensor model proposed in \cite{Witten:2016iux} in Euclidean time is:
\begin{align*}
 S = \int d \tau \left[ \frac{1}{2 } \sum_c \left( \sum_{ {\bf a}^c } \psi^c_{{\bf a}^c} \frac{d}{d\tau} \psi^c_{{\bf a}^c} \right) -
  \imath^{(D+1)/2} \frac{J}{n^{\frac{D(D-1)}{4}}}  \sum_{ {\bf a}^0 , \dots, {\bf a}^D  } \psi^0_{{\bf a}^0} \dots \psi^D_{{\bf a}^D} \prod_{c_1<c_2} \delta_{a^{c_1c_2} a^{c_2c_1}}
 \right] \;. 
\end{align*}
This model is very similar to the colored tensor model (CTM) \cite{review,color}, except that:
\begin{itemize}
 \item[-] The SYK--like model of \cite{Witten:2016iux} is fermionic. With the exception of the first paper on the topic, \cite{color}, the CTM is usually formulated for bosonic fields. 
 The statistics of the field has no bearing on the structure of the $1/n$ expansion \emph{i.e.} the class of graphs contributing at a given order in $1/n$.
 \item[-] The SYK--like model of \cite{Witten:2016iux} has a non trivial propagator ($\imath/p$ in momentum space). While this propagator encodes the entire physical interpretation of the model, it again has no bearing on the $1/n$ expansion.
 \item[-] Finally, the SYK--like model of \cite{Witten:2016iux} is formulated in terms of real fields, while the CTM is usually formulated in terms of complex fields. This slightly changes the $1/n$ expansion and we will briefly
 allude to this at the end of this paper.
\end{itemize}

Given the similarity of the two models, it is straightforward to translate the full $1/n$ expansion of the CTM \cite{GurSch} to the SYK--like model of \cite{Witten:2016iux}.
For convenience we will be considering in the rest of this paper the complex version of the model: 
\begin{align*}
 S = \int d \tau \bigg[ \sum_c \left( \sum_{ {\bf a}^c } \bar \psi^c_{{\bf a}^c} \frac{d}{d\tau} \psi^c_{{\bf a}^c} \right) & -
   \frac{J}{n^{\frac{D(D-1)}{4}}}  \sum_{ {\bf a}^0 , \dots, {\bf a}^D  } \psi^0_{{\bf a}^0} \dots \psi^D_{{\bf a}^D} \prod_{c_1<c_2} \delta_{a^{c_1c_2} a^{c_2c_1}} \crcr
 & - \frac{J}{n^{\frac{D(D-1)}{4}}}  \sum_{ {\bf a}^0 , \dots, {\bf a}^D  }  \bar \psi^D_{{\bf a}^D} \dots \bar \psi^0_{{\bf a}^0} \prod_{c_1<c_2} \delta_{a^{c_1c_2} a^{c_2c_1}}
 \bigg] \;,
\end{align*}
and we only discuss the case with at least four colors, $D\ge 3$.

\section{Discussion}

Before proceeding to the technical part of the paper, let us comment in more detail on the results we obtain and their interpretation.
In \cite{Witten:2016iux} Witten observed that the leading order two point function of a SYK--like tensor model obeys a melonic Schwinger Dyson equation:
\[
  \frac{1}{ G_{\rm LO} (\omega) } = - \imath \omega - \Sigma(\omega) \;, \qquad \Sigma(\tau,\tau') = J^2 [ G_{\rm LO}(\tau,\tau')]^D \;,
\]
identical to the one obtained in the usual SYK model  \cite{Polchinski:2016xgd,Jevicki:2016bwu,Maldacena:2016hyu}. The main point is that, contrary to the 
usual SYK models, this equation is obtained in tensor SYK models without resorting to a quench disorder average.

Seeing that the two point function in the SYK--like tensor model has this behavior, two questions arise naturally:
\begin{itemize}
 \item what can be said about the four point function(s) in the SYK--like tensor model? It is known that in the usual SYK model 
 \cite{Polchinski:2016xgd,Jevicki:2016bwu,Maldacena:2016hyu} the relevant contribution to the four point function 
 is given by the sum over four point ladder diagrams. Does something similar hold for the SYK--like tensor model? 
 \item what can be said about the $1/n$ series of the SYK--like tensor model? More precisely, how do 
 the results on the full $1/n$ expansion of tensor models in \cite{GurSch} translate for the SYK--like tensor model? 
\end{itemize}

The answers to these questions are more subtle than it might appear at first sight. 

\paragraph{Four point functions.}

It is not a priory obvious that if melons dominate the two point function then ladders will dominate 
the four point function. In fact this is not true. The ladder diagrams play a distinguished role in the SYK--like tensor model, but for 
a more subtle reason. 

Four point functions per se have not been studied so far in the tensor model literature. While straightforward, their analysis is performed 
for the first time here.  

Due to the fact that the fields are tensors, one needs to proceed with care.
As a function of the pattern of identification of the external indices, one distinguishes several four point 
functions (see Section~\ref{ssec:fourpoint}): \emph{exceptional}, \emph{broken}, \emph{unbroken} and \emph{small}. Each of these four point functions will have 
a certain leading order scaling in $1/n$. It is however naive to just
compare the scalings and decide that some four point functions are subleading with respect to the others (by this 
argument one would conclude that the exceptional four point functions 
dominate over all the others). This is incorrect because each type of four point function corresponds to
a different patterns of identification of the external indices. The correct way to analyze the scaling with $n$ of the 
four point functions is to contract the external indices either among themselves or on new vertex kernels and consider 
the scaling with $n$ of this contraction. From this analysis we conclude that the small four point functions are suppressed in powers of $1/n$.

In Section~\ref{ssec:fourpointlead} we show that the ladder (or chain) diagrams yield the leading order contributions to
the broken and unbroken four point functions. In particular, the ``out of time order'' four point correlation which is crucial for the 
usual SYK models \cite{Maldacena:2016hyu} corresponds precisely to broken and unbroken four point functions in the tensor case.
The exceptional four point functions are not subleading with respect to the broken and unbroken ones, but, due to the pattern of identification of their 
external indices, do not contribute to the ``out of time order'' correlation.

\paragraph{The full $1/n$ series and non perturbative effects.} 
In Section~\ref{sec:twopointallorders} we discuss the full $1/n$ series of the tensor SYK model. The results we preset 
classify the non perturbative effects one encounters at any order in $1/n$.

In the usual SYK models (as well as in the tensor SYK models) non perturbative effects play a crucial role. The chains (ladder diagrams) are a sufficiently simple family of
graphs and can be resummed analytically \cite{Sachdev:1992fk,Kitaev,Maldacena:2016hyu}. They are a  geometric series in an elementary four 
point kernel $K$ (using the notation of \cite{Maldacena:2016hyu}) leading to resummed four point function $(1-K)^{-1}$. The non triviality of the four point function (ultimately leading to a saturation of the chaos bound 
\cite{Maldacena:2016hyu,Maldacena:2015waa})
comes from the fact that the elementary kernel $K$ has eigenvalue $1$. 

However, non perturbative effects in usual quantum field theory and matrix models arise from infinite families of graphs which are much more complicated
than chains: the \emph{parquet graphs}. Although parquet graphs arise in any quantum field theory, for simplicity let consider the case of a quartic interaction. 
The lowest order radiative correction to the 
coupling constant is given by the four point ``bubble graph'' consisting in two vertices connected by two edges 
(hence with four external half edges). For any graph of the model, one can build an infinite family of parquet
graphs by inserting iteratively in all possible manners the bubble graph at any vertex. 
Parquet graphs play a crucial role in QCD \cite{'tHooft:1973jz}  for instance:
in the planar limit of QCD one would in principle need to resum the family of planar parquet 
graphs and it is only after this resummation that 
one can access non perturbative effects like quark confinement \cite{'tHooft:1973jz}.

The problem with parquet graphs is that they are quite difficult to handle: to this day, there are no viable methods 
to resum them analytically. When studying subleading orders in the $1/n$ expansion of the SYK and tensor SYK models, one 
runs the risk of having to resum such infinite families of parquet graphs: if this were to happen, the analytic computation would be 
impossible.

So, does one need to resum the full family of parquet graphs at some order in $1/n$ in the tensor SYK model? 
The main result of this paper is that the answer to this question is \emph{no}.
Theorem~\ref{thm:mainmic} in Section~\ref{sec:twopointallorders} states that any given order in the $1/n$ series in the tensor SYK model is a finite sum of convolutions of a finite number of 
leading order two and four point (broken and unbroken) kernels. 

As it happens, these are exactly the kernels one needs to compute in order 
to estimate the ``out of time order correlation'' and test the chaos bound \cite{Maldacena:2016hyu}.
It should be stressed that, while it is natural for the leading order two point function to appear, that there is no a priori reason why only the leading order 
chain diagrams should contribute to any order in $1/n$. This is very different from what happens in matrix models for instance. Furthermore, there is no a priori reason
why the number of convolutions at any order in $1/n$ should be finite. Again, this does not happen in matrix models.
Proving these two statements is quite involved \cite{GurSch}. However, the fact that both these statements are true 
strongly restricts the non perturbative effects at a fixed order in $1/n$:
the only non perturbative effects one encounters at any fixed order in $1/n$ come from the resummation of the chain 
diagrams as geometric series in the elementary four point kernel $K$.
While Theorem~\ref{thm:mainmic} is a direct consequence of the parallel result obtained in tensor models \cite{GurSch}, it has
drastic consequences for the non perturbative effects in SYK--like tensor models.

\section{Two and four point functions}

We will discuss in detail the connected two and four point functions of the model. General $2p$ point functions can be obtained by similar methods.

\subsection{The two point function}

From symmetry arguments it follows that the connected two point functions of the model are non zero only if the fields have the same color and the same indices:
\[ \Braket{ \psi^c_{ {\bf a}^c } (\tau) \bar \psi^{c'}_{ {\bf b}^{c'} } (\tau')  }_{\rm c} = \frac{  \int [d\bar \psi d\psi] \;    \psi^c_{ {\bf a}^c }(\tau) \bar \psi^c_{ {\bf b}^{c'} } (\tau') e^{-S} }
  {  \int [d\bar \psi d\psi] \;   e^{-S} }   = G(\tau,\tau') \left( \delta^{cc'}\prod_{c_1\neq c} \delta_{a^{cc_1} b^{c'c_1}} \right) \;,
\]
where the \emph{normalized two point function} $G(\tau,\tau')$:
\[
  G(\tau,\tau') =    \frac{1}{n^D}  \; \frac{ \int [d\bar \psi d\psi] \;  \left[ \sum_{ {\bf a}^c }  \psi^c_{ {\bf a}^c } (\tau) \bar  \psi^c_{ {\bf a}^c } (\tau') \right] e^{-S} }  {  \int [d\bar \psi d\psi] \;   e^{-S} } 
   = \sum_{\cG}  {\bf A}^{\cG}(\tau,\tau') \;,
\]
is a sum of connected graphs $\cG$ (see Fig.~\ref{fig:examples}) with the following properties:
\begin{itemize}
 \item[-] the graphs are bipartite (which we track by coloring the vertices black and white) and every edge has a color $c\in \cC$.
 \item[-] every vertex is incident to exactly one edge for each color.
 \item[-] one edge (representing the insertions $ \sum_{ {\bf a}^c }  \psi^c_{ {\bf a}^c } (\tau) \bar  \psi^c_{ {\bf a}^c } (\tau')$ in the functional integral) is 
 marked\footnote{The two point functions with fixed external arguments are sums over graphs with two external half edges. The sums over the indices of the insertions 
 in $G(\tau, \tau')$ represent the gluing of the two half edges into an edge.}. We call this edge the \emph{root} of $\cG$. 
\end{itemize}
In the rest of this paper we call such graphs \emph{closed} (as they don't have external half edges) for short. Opening the root edge of a closed graph into two half edges leads to a
two point graph.

 \begin{figure}[htb]
 \begin{center}
 \includegraphics[width=6cm]{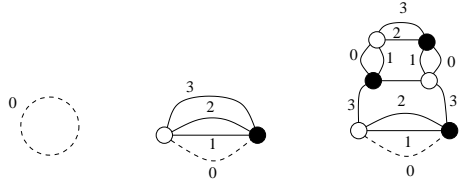}  
 \caption{Closed, bipartite, edge $3+1$ colored, rooted graphs. The root is represented as dashed.} \label{fig:examples}
 \end{center}
 \end{figure}
 
 From the onset we include among the closed graphs the \emph{ring} graph represented by and edge closing onto itself and having no vertex (Fig.~\ref{fig:examples} on the left). It corresponds to the bare propagator 
 contribution to the two point function.

As remarked in \cite{Witten:2016iux}, the amplitude of $\cG$ factors into a pure scaling with $n$ and an integral over the internal positions of the vertices:
\begin{equation}\label{eq:ampli}
 {\bf A}^{\cG}(\tau,\tau') = n^{-\frac{2}{(D-1)!} \omega(\cG) } A^{\cG}(\tau,\tau') \;,
\end{equation}
where $\omega(\cG)$ is called the \emph{degree} of $\cG$ \cite{review}. This formula is obtained \cite{review} by noting that:
\begin{itemize}
 \item every vertex contributes $n^{-\frac{D(D-1)}{4}}$. We denote $2k(\cG)$ the number of vertices of $\cG$ (which is even as $\cG$ is bipartite).
 \item every sum over an index $a^{c_1c_2}$ brings a factor $n$. An index $a^{c_1c_2}$ is identified along the edges of colors $c_1$ and $c_2$ of the graph, and 
 one free sum is obtained for every closed cycle with colors $c_1c_2$. These closed cycles are called the \emph{faces} with colors $c_1c_2$ of $\cG$ (and we 
 denote $F^{c_1c_2}(\cG)$ their number).
\end{itemize}
Taking into account the overall factor $1/n^D$, the scaling with $n$ of a graph is first obtained as: 
\[n^{ - D -\frac{D(D-1)}{2} k(\cG) + \sum_{c_1<c_2} F^{c_1c_2} (\cG)} \;.
\]

A cyclic permutation $\pi$ over the colors induces an embedding of $\cG$ into a two dimensional surface of minimal genus $g_{\pi}(\cG)$ by ordering the edges clockwise (resp. counterclockwise) around the 
white (resp. black) vertices in the order $0,\pi(0), \pi^2(0) \dots \pi^D(0)$. The two cells of the embedding (\emph{i.e.} the faces of the embedding) are exactly the faces with colors 
$\pi^q(0)\pi^{q+1}(0)$ for $q=0,\dots D$. For every embedding the Euler relation reads $ 2-2g_{\pi} (\cG) = 2k(\cG) - (D+1) k(\cG)  + \sum_{q=0}^D F^{\pi^q(0)\pi^{q+1}(0)} (\cG) $.
Taking into account that there are $D!$ cycles $\pi$, and that the faces with colors $c_1c_2$ are counted by $2(D-1)!$ embeddings, summing over $\pi$ we obtain:
\begin{align}\label{eq:facesdeg}
 & 2 D ! -2 \sum_{\pi} g_{\pi}(\cG) = D! (D-1)k(\cG) + 2 (D-1)! \sum_{c_1<c_2} F^{c_1c_2}(\cG) \Rightarrow \crcr
 & \qquad  \sum_{c_1<c_2} F^{c_1c_2} (\cG ) = D + \frac{D(D-1)}{2} k (\cG) - \frac{2}{(D-1)!} \left( \frac{1}{2} \sum_{\pi} g_{\pi}(\cG) \right) \;,
\end{align}
and Eq.~\eqref{eq:ampli} follows with $ \omega(\cG) = \frac{1}{2} \sum_{\pi} g_{\pi}(\cG)  $ (which is an integer as the embeddings $\pi$ and $\pi^{-1}$ are identified by mirror symmetry).
In this form it is self evident that $\omega(\cG)$ is a non negative integer and indexes the $1/n$ series of the SYK--like model of \cite{Witten:2016iux}.

The $3+1$ colored graphs in Fig.~\ref{fig:examples} have degrees, from left to right: $0$, $0$ and $1$ (the ring graph with $D+1$ color has $D$ faces, one for each color different from the color of its unique edge).

\subsection{The four point functions}\label{ssec:fourpoint}

The four point functions with fixed external arguments are sums over connected graphs with 4 external half edges. As a function of $D$, several classes of four point functions 
are generated. Simple combinatorial arguments and conservation of the external 
indices along bi colored paths of edges lead to the following classification of the four point functions:
\begin{itemize}
     \item For $D=3$ two exceptional four point functions (represented in Fig.~\ref{fig:exceptional}): 
              \begin{align*}
               & \Braket{ \psi^0_{{ \bf a}^0 } (\tau_0) \psi^1_{{ \bf a}^1 }(\tau_1)\psi^2_{{ \bf a}^2 }(\tau_2)\psi^3_{{ \bf a}^3 } (\tau_3)}_{\rm c} 
                    = F_{\rm exceptional} (\tau_0,\tau_1,\tau_2,\tau_3) \prod_{c_1<c_2} \delta_{a^{c_1c_2} a^{c_2c_1}} \;, \crcr 
               & \Braket{ \bar \psi^0_{{ \bf a}^0 } (\tau_0) \bar \psi^1_{{ \bf a}^1 }(\tau_1)  \bar \psi^2_{{ \bf a}^2 }(\tau_2) \bar \psi^3_{{ \bf a}^3 } (\tau_3)}_{\rm c} = 
               \bar F_{\rm exceptional} (\tau_0,\tau_1,\tau_2,\tau_3) \prod_{c_1<c_2} \delta_{a^{c_1c_2} a^{c_2c_1}} \; . 
               \end{align*}
  \begin{figure}[htb]
 \begin{center}
 \includegraphics[height=4cm]{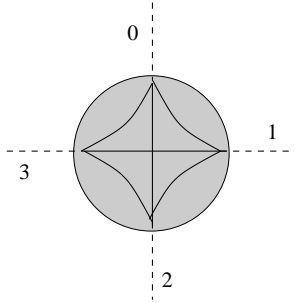}  
 \caption{Exceptional four point function in $D=3$. The solid lines indicate the identifications of external indices.} \label{fig:exceptional}
 \end{center}
 \end{figure}
        The exceptional four point functions have four different external colors.
     \item For any $D \ge 3$, \emph{broken} and \emph{unbroken} for point functions  (represented in Fig.~\ref{fig:brokenunbrokesmall}), and for $D\ge 4$ and only one external color also \emph{small} four point functions:
        \begin{align*}
         & \Braket{ \psi^0_{{ \bf a}^0 } (\tau_0) \bar \psi^0_{{ \bf b}^0 }(\tau_1) \psi^1_{{ \bf p}^1 }(\tau_2) \bar \psi^1_{{ \bf q}^1 } (\tau_3)}_{\rm c} = \crcr 
         & \;  =     F^{ \neq}_{\rm B}(\tau_1,\tau_2,\tau_3,\tau_4) 
              \left( \prod_{c\neq 0} \delta_{a^{0c} b^{0c} }\right) \left( \prod_{c'\neq 1} \delta_{p^{1c'} q^{1c'} }\right)   + \crcr   
        & \; \qquad  +  F^{\neq  }_{\rm U}(\tau_0,\tau_1,\tau_2,\tau_3)  \; \delta_{a^{01} p^{10}}  \delta_{b^{01} q^{10}} 
             \left( \prod_{c\neq 0,1}\delta_{a^{0c} b^{0c}} \right)           \left( \prod_{c'\neq 0,1}\delta_{p^{1c'} q^{1c'}} \right)         
           \crcr
          & \Braket{ \psi^0_{{ \bf a}^0 } (\tau_0) \bar \psi^0_{{ \bf b}^0 }(\tau_1) \psi^0_{{ \bf p}^0 }(\tau_2) \bar \psi^0_{{ \bf q}^0 } (\tau_3)}_{\rm c} =  \crcr
          & \; =  F^{ = }_{\rm B}(\tau_1,\tau_2,\tau_3,\tau_4)  \left[ \left( \prod_{c\neq 0} \delta_{a^{0c} b^{0c} }\right) \left( \prod_{c'\neq 0} \delta_{p^{0c'} q^{0c'} }\right)  
               -  (b\leftrightarrow q )  \right]   + \crcr                 
         & \; \qquad +  F^{= }_{\rm U}(\tau_0,\tau_1,\tau_2,\tau_3)  \; \left[ \sum_c \delta_{a^{0c} p^{0c}}  \delta_{b^{0c} q^{0c}} 
             \left( \prod_{c '\neq c}\delta_{a^{0c'} b^{0c'}} \right)           \left( \prod_{c'\neq c}\delta_{p^{0c'} q^{0c'}} \right)   -    (b\leftrightarrow q )    \right]   \crcr 
             & \; \qquad +  F^{= }_{\rm small}(\tau_0,\tau_1,\tau_2,\tau_3)  \; \left[ \sum_{ \cC' \subset \cC}^{ 2\le |\cC'| \le [D/2] } 
                \left(  \prod_{c \in \cC' } \delta_{a^{0c} p^{0c}}  \delta_{b^{0c} q^{0c}} \right)
             \left( \prod_{c \notin \cC' } \delta_{a^{0c} b^{0c}}  \delta_{p^{0c} q^{0c}} \right)   -    (b\leftrightarrow q )    \right] \;,
        \end{align*}
        where the upper index $=$ or $\neq$ indicates that the external half edges either all have the same color, or are paired into two pairs with different colors.
        The small four point functions appear only in $D\ge 4$ and only if the four external colors are all equal.
  \begin{figure}[htb]
 \begin{center}
 \includegraphics[width=12cm]{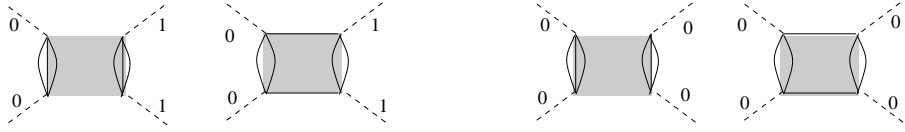}  
 \caption{Broken and unbroken four point functions with different (resp. equal) external colors for $D=3$.} \label{fig:brokenunbrokesmall}
 \end{center}
 \end{figure}
\end{itemize}

By symmetry under relabeling of the colors, analyzing the exceptional four point functions, broken and unbroken four point functions with different external colors
and broken, unbroken and small four point functions with equal external colors suffices.
We will see below that the most relevant four point functions are the broken and unbroken ones with either equal or different external colors.

\section{The leading order in $1/n$}\label{sec:leading1/n}

Two classes of graphs play an important role below: the melonic graphs and the chains. Two equivalent definitions of melonic graphs can be found in the literature \cite{GurSch,critical},
and we present and use them both.

\paragraph{First definition of two point melonic graphs \cite{critical}.}

The first example of a two point melonic graph, the \emph{trivial two point graph} consists in an edge and no vertex (see Fig.~\ref{fig:example1}, on the left). 
The second example is the \emph{fundamental melon} (see Fig.~\ref{fig:example1}, in the middle), that is the two point graph 
with exactly two vertices. Due to the coloring constraints the vertices are connected by $D$ edges, and the two external half 
edges have the same color. 
 \begin{figure}[htb]
 \begin{center}
 \includegraphics[width=10cm]{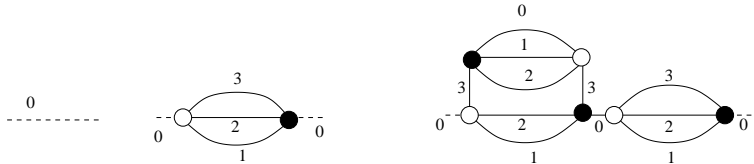}  
 \caption{First definition of two point melonic graphs.} \label{fig:example1}
 \end{center}
 \end{figure}

All the other two point melonic graphs are obtained by inserting the fundamental melon arbitrarily on the edges (including the right external half edge) of two point melonic graphs 
(an example is presented Fig.~\ref{fig:example1}, on the right).

\paragraph{Second definition of two point melonic graphs \cite{GurSch}.}
Consider a two point graph $\cG$, \emph{i.e.} a graph with two external half edges (of the same color). We define the (possibly empty) \emph{root cut set} of $\cG$ as the 
set of all the one particle reducibility edges in $\cG$ (\emph{i.e.} the maximal set of edges $f\in \cG$ such that $\cG$ splits into two connected components
by cutting $f$). All the edges in the root cut set have the color of the external half edges of $\cG$.
Then $\cG$ is a:
\begin{itemize}
 \item \emph{two point melonic graph} if all the connect components obtained by cutting the edges in the root cut set are \emph{two point prime melonic graphs} (or it is the trivial two point graph).
 \item \emph{two point prime melonic graph} if by deleting the two vertices hooked to the external half edges and cutting all the edges incident to them, the graph splits in 
 exactly $D$ (possibly trivial) \emph{two point melonic graphs}.
\end{itemize}

This definition is illustrated in Fig.~\ref{fig:meldef} below.
 \begin{figure}[htb]
 \begin{center}
 \includegraphics[width=6cm]{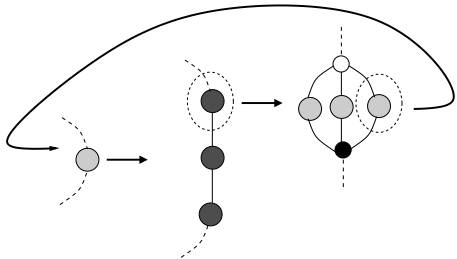}  
 \caption{Second definition of two point melonic graphs.} \label{fig:meldef}
 \end{center}
 \end{figure}

 It is a simple exercise \cite{GurSch} to show that the two definitions are equivalent.
 
 \paragraph{Closed melonic graphs.} A graph obtained by connecting the two external half edges of a two point melonic graph into a root edge
is called a \emph{closed melonic graph}. Observe that the ring graph is obtained by closing the trivial two point graph, and by convention is considered closed melonic.

\paragraph{Four point chain graphs.} A $(D-1)$-dipole with external colors $c_1$ and $c_2$  (Fig.~\ref{fig:example2} on the left) is a four point graph with exactly two vertices. 
Due to the coloring constraints the two vertices are connected by $D-1$ edges. A $(D-1)$-dipole has two pairs of external half edges, the right half edges with color $c_1$ and the left
half edges with color $c_2$. Dipoles can join together to form chains (Fig.~\ref{fig:example2} in the middle and on the right). This is done by connecting the 
left half edges of a dipole to the right half edges of another dipole respecting the coloring.

 \begin{figure}[htb]
 \begin{center}
 \includegraphics[width=10cm]{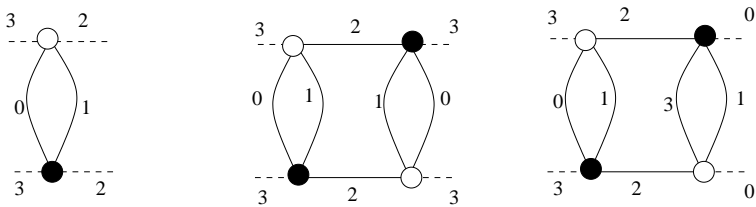}  
 \caption{A dipole, an unbroken and a broken chain.} \label{fig:example2}
 \end{center}
 \end{figure}

 Chains are either \emph{unbroken} (Fig.~\ref{fig:example2} in the middle), if two faces are transmitted along the chain from the left to the right half edges (one along the top and one along the bottom
 of the chain), or \emph{broken} if no face is transmitted (Fig.~\ref{fig:example2} on the right). Furthermore the chains can have either equal or different external colors.

\subsection{Two point function at leading order}

The degree of a closed graph does not change when inserting a two point melonic graph on an edge. This can readily be seen from Eq.~\eqref{eq:facesdeg}, as the insertion of a fundamental melon 
 on an edge brings $2$ new vertices and $\frac{D(D-1)}{2}$ new faces (recall that $k(\cG)$ is the half number of vertices of $\cG$).
 In particular all the closed melonic graphs have degree $0$, as the ring graph (which is closed and melonic) has $D$ faces (the edge has some color $c$ and we get a face for all the colors $\cC \setminus \{c\}$) 
 and no vertex, hence degree $0$. The converse statement is also true.
\begin{proposition}\cite{critical}\label{prop:melons}
 For $D\ge 3$ a closed graph has degree zero if and only if it is a closed melonic graph. 
\end{proposition}
The proof of this statement \cite{critical} is somewhat convoluted.
The first remark is that, if one denotes $F_s(\cG)$ the total number of faces with $2s$ vertices of $\cG$, one has on the one hand:
\[ \sum_{s\ge 1} F_s(\cG) = D + \frac{D(D-1)}{2}k(\cG) - \frac{2}{(D-1)!} \omega(\cG)  \;, \]
and on the other $\sum_{s} (2s) F_s(\cG) = D(D+1)k(\cG)$, as every edge belongs to exactly $D$ faces and $\cG$ has $(D+1)k(\cG)$ edges. 
Eliminating $k(\cG)$ one obtains:
\begin{align}\label{eq:truc}
 (D+1) \frac{2}{(D-1)!} \omega(\cG) + 2F_1(\cG) = D(D+1) + \sum_{s\ge 2} [s(D-1) - D -1 ] F_s(\cG) \;,
\end{align}
and $s(D-1) - D -1 \ge D-3 $ for $s\ge 2$, i.e. all the coefficients of the sum in the right hand side are non negative. 

It follows that for $D\ge 3$ (in contrast to $D=2$) any graph with degree zero has faces with exactly two vertices.
As, furthermore, all the $D!$ embeddings of a graph with degree $0$ are planar, a short induction \cite{critical} proves Proposition \ref{prop:melons}.

The normalized two point function scales like $G(\tau,\tau')  \sim n^0$ and at leading order equals:
 \[G_{\rm LO}(\tau,\tau') = \sum_{ \cM \in \genfrac{}{}{0pt}{}{\text{two point melonic graphs}}{\text{with external color $0$} } } A^{\cM}(\tau,\tau') \; , \]
where we used the fact that rooted closed melonic graphs and two point melonic graph are in bijection: while in order to determine the scaling with $n$ it is more convenient 
to use the closed graphs, once the melonic family is singled out we go back to a more familiar sum over two point melonic graphs. 

Before concluding, let us note that Eq.~\eqref{eq:truc} tells us that not only a graph with degree zero has faces with exactly two vertices, 
but also that, for any graph, the number of large faces (with more than two vertices)\footnote{A technicality is that for $D=3$, the coefficient of 
$F_2(\cG)$ is exactly zero, hence the number of faces with four vertices is not bounded by this equation. A detailed analysis \cite{GurSch} shows that the case $D=3$ indeed behaves like the case $D\ge 4$.}
and the maximal number of vertices of a face
are bounded linearly in the degree and the number $F_1(\cG)$ of short faces. This will prove crucial later on.

The two point melonic graphs are identical with the leading order graphs for the two point function in the standard SYK model \cite{Polchinski:2016xgd,Jevicki:2016bwu,Maldacena:2016hyu}.
 
\subsection{Four point functions at leading order}\label{ssec:fourpointlead}

We have several cases.

\begin{description}
 \item[Exceptional four point functions.] They arise only in $D=3$ and have external colors $0,1,2$ and $3$. Let us contract an exceptional four point function with a vertex kernel. For any graph contributing to this four point function 
 we obtain a closed connected graph. As closed graphs scale like $n^{3-\omega(\cG)}$, we can derive the maximal scaling of exceptional four point functions (recall that a vertex brings a factor $n^{-\frac{3\cdot 2}{4}}$ and a face a 
 factor $n$):
 \begin{align*}
  F_{\rm exceptional} (\tau_0,\tau_1,\tau_2,\tau_3) n^{- \frac{3\cdot 2}{ 4} + 6} = n^3 \Rightarrow F_{\rm exceptional} (\tau_0,\tau_1,\tau_2,\tau_3) \sim n^{-\frac{3}{2}} \;.
 \end{align*}
  \begin{figure}[htb]
 \begin{center}
 \includegraphics[height=3cm]{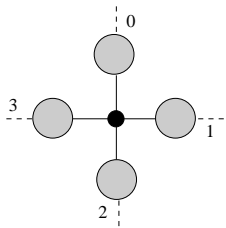}  
 \caption{Leading order exceptional four point functions in $D=3$.} \label{fig:LOexceptional}
 \end{center}
 \end{figure}
 
   Observe that a bare vertex decorated by leading order two point functions on its half edges (see Fig.~\ref{fig:LOexceptional}) becomes melonic by contracting it with a vertex kernel, hence reproduces 
   the scaling $n^{-\frac{3}{2}}$. The converse is also true.
  \begin{proposition}
   A four point graph contributes to the leading order of an exceptional for point function if and only if it is a bare vertex decorated by melonic two point graphs on its half edges.
  \end{proposition}
  \begin{proof}
  A graph with four external half edges contributes to the leading order exceptional four point function if and only if, when adding a vertex and connecting
  the half edges respecting the coloring, the resulting closed graph is melonic.
  
  Any vertex in a closed melonic graphs has a canonical partner such that when deleting the two vertices the melonic graph splits into $D+1$ melonic two point graphs. 
  This is obvious when considering the first definition of melonic graphs: a canonical pair is formed by 
  the two vertices introduced by the insertion of a fundamental melon. It follows that the leading order contribution to the 
  exceptional four point functions is the bare vertex decorated by a leading order two point function on each of its half edges.   
  \end{proof}
  
\item[Broken and unbroken four point functions.] They arise for all $D$ and have either equal or different external colors (say $0,0,0,0$ or $0,0,1,1$). 
Let us first consider the case with different external colors ($0,0,1,1$). We contract the indices of the color $0$ (resp. color $1$) fields.
We obtain a closed colored graph which scales at most like $n^D$, and it scales like $n^D$ only if it is melonic.
It follows that:
\begin{align*}
   & F^{ \neq}_{\rm B}(\tau_1,\tau_2,\tau_3,\tau_4) n^{2D} \sim  F^{\neq  }_{\rm U}(\tau_0,\tau_1,\tau_2,\tau_3) n^{ 2D-1 } \sim n^D  \Rightarrow \crcr
   & \qquad F^{ \neq}_{\rm B}(\tau_1,\tau_2,\tau_3,\tau_4) \sim n^{-D} \; ,\qquad  F^{\neq  }_{\rm U}(\tau_0,\tau_1,\tau_2,\tau_3) \sim n^{-D+1} \;.
\end{align*}

The broken and unbroken chains decorated by arbitrary melonic two point graphs (see Fig.~\ref{fig:LObrokenunbroken}) reproduce these scalings as they become melonic when reconnecting the left and right half edges into edges.
  \begin{figure}[htb]
 \begin{center}
 \includegraphics[height=3cm]{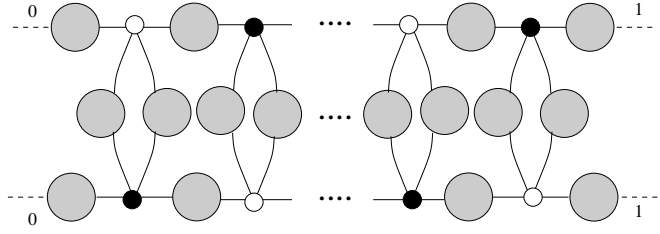}  
 \caption{Leading order broken (or unbroken) four point functions.} \label{fig:LObrokenunbroken}
 \end{center}
 \end{figure}
The converse statement is also true.
\begin{proposition}\label{prop:chains}
A graph contributes to the leading order broken (resp. unbroken) four point functions if and only if it is a broken (resp. unbroken) chain decorated by melonic two point graphs on its edges.
\end{proposition}
\begin{proof}
 A graph $\cG$ with two half edges of color $1$ and two half edges of color $0$ contributes to the leading order four point function only if, by connecting the external half edges 
 of color $1$ into an edge $e^1$ it becomes a two point melonic graph $\cG'$. The edge $e^1$ does not belong to the cut set of the root of $\cG'$ (because $\cG'$ comes from a connected four point function,
 hence can not be disconnected by cutting $e^1$). It follows that $e^1$ will belong to one of the connected components (which are all two point prime melonic graphs), say $\cG''$ obtained by cutting the edges 
 in the cut set of the root of $\cG'$. Now, deleting the end vertices of $\cG''$, one obtains $D$ two point melonic graphs and $e^1$ belongs to one of them, say $\cG'''$. 
 The two end vertices of $\cG''$, the melonic graphs different from $\cG'''$ separated by deleting them and the connected components generated by cutting the edges in the cut set of the root 
 of $\cG$ identify a first $(D-1)$-dipole decorated by melonic two point graphs separating the half edges of color $0$ from the half edges of color $1$ in $\cG$. We conclude by induction.
\end{proof}

With minimal adaptation, a similar statement holds for the broken and unbroken contributions to the four point function with equal external colors. 
The leading order broken and unbroken four point functions are sums over chain (ladder) diagrams with propagators the leading order two point function:
\begin{align*}
& F^{\neq}_{\rm U;LO}(\tau_1,\tau_2,\tau_3,\tau_4) = \sum_{ {\cal U}  \in \genfrac{}{}{0pt}{}{ \text {unbroken chains with}}{\text{different external colors}}} \hat A^{ {\cal U} ;\neq}(\tau_1,\tau_2,\tau_3,\tau_4)  \crcr
& F^{=}_{\rm U;LO}(\tau_1,\tau_2,\tau_3,\tau_4) = \sum_{ {\cal U}  \in \genfrac{}{}{0pt}{}{ \text {unbroken chains with}}{\text{equal external colors}}} \hat A^{ {\cal U} ; = }(\tau_1,\tau_2,\tau_3,\tau_4)  \crcr
& F^{\neq}_{\rm B;LO}(\tau_1,\tau_2,\tau_3,\tau_4) = \sum_{ {\cal U}  \in \genfrac{}{}{0pt}{}{ \text {broken chains with}}{\text{different external colors}}} \hat A^{ {\cal B} ;\neq}(\tau_1,\tau_2,\tau_3,\tau_4)  \crcr
& F^{ = }_{\rm B;LO}(\tau_1,\tau_2,\tau_3,\tau_4) = \sum_{ {\cal U}  \in \genfrac{}{}{0pt}{}{ \text {broken chains with}}{\text{equal external colors}}} \hat A^{ {\cal B} ;=}(\tau_1,\tau_2,\tau_3,\tau_4) \; , \crcr
\end{align*}
 where the hat on $\hat A$ signals that the amplitude is evaluated using as propagator the leading order two point function $G_{\rm LO}(\tau,\tau')$.

\item[Small four point functions.] They arise only for $D\ge 4$ and equal external colors. By contracting the external indices of a four point graph contributing to a small four point function one can not build a melonic
 graph (as only broken and unbroken chains reconstitute melonic graphs), hence these four point functions are power counting suppressed.
\end{description}

The chains decorated by melonic graphs on their edges are identical with the leading order ladder graphs for the four point functions in the standard SYK model \cite{Polchinski:2016xgd,Jevicki:2016bwu,Maldacena:2016hyu}.

\section{The complete $1/n$ series of the two point function}\label{sec:twopointallorders}

The leading order two point function and the leading order broken and unbroken four point functions allow one to 
compute any fixed order in $1/n$ of the two point function. The same holds for arbitrary correlations, but the details of the proof are left to 
the reader.

For any $\omega$, either there is no closed graph with degree $\omega$ or there is an infinite family of closed graphs with degree $\omega$. 
This is due to two facts:
\begin{itemize}
 \item as we have already seen, the degree does not change when inserting a two point melonic graph on an edge.
  \item the degree is insensitive to the length of the chains. Let us call \emph{internal faces} of a chain the faces made by edges of the chain which do \emph{not pass} through the external half edges.
  For example in Fig.~\ref{fig:example2} in the middle, one counts two  internal faces with colors (0,1), one with colors (0,2)  and one with colors (1,2). The other faces, with colors
  (0,3), (1,3) and (2,3) pass through the external half edges and are not internal faces of the chain. For the example in  Fig.~\ref{fig:example2} on the right, a face (0,1) is internal, while the
  second face (0,1) is not.
   
  The number of internal faces and internal vertices of the chain changes with its length. The number of internal faces of a broken chain ${\cal B}$, respectively unbroken chain ${\cal U}$
  is:
  \begin{align*}
      F^{\rm int}({\cal B} ) + 2D = D +  \frac{D(D-1)}{2} k ({\cal B})   \;, \qquad  F^{\rm int}({\cal U} ) + 2D -1 = D +  \frac{D(D-1)}{2} k ( {\cal U} )  \;.
  \end{align*}
because gluing the left (resp. right) half edges together one obtains a closed melonic graph. 
 It follows that the variation of the number of faces of a graph when the length of an internal chain varies
   is always $ \Delta F(\cG) = \Delta F^{\rm int} ({\cal U}) = \Delta F^{\rm int} ({\cal B}) = \frac{D(D-1)}{2} \Delta k(\cG)$ and the degree is unaltered.
   \end{itemize}

The point is that these are the \emph{only reasons} for which the number of graphs with a given degree is infinite \cite{GurSch}.  

\paragraph{The core of a graph.}

Melonic two point subgraphs are either totally disjoint, or their union is a melonic two point subgraph.
A similar property holds for arbitrary two point subgraphs, and the reader just needs to check that if each of the two point subgraphs 
is melonic, then their union is also melonic \cite{GurSch}. This allows one to define \emph{maximal} melonic subgraphs (\emph{i.e.} melonic subgraphs
which are maximal for inclusion). 

For any closed graph $\cG$, the \emph{core} $\hat \cG$ of $\cG$ is obtained by replacing all the maximal melonic subgraphs by edges.
However, the number of cores at fixed  degree is infinite, as a core can have chains of arbitrary length.
By construction a core is melon free, that is it has no melonic subgraph.

\paragraph{The reduced scheme of a core.}

The case $D=3$ is special, and it is the root of much of the technical difficulties in \cite{GurSch}. However, at the end, the classification of graphs at $D=3$ is identical to
the one of graphs with $D\ge 4$. We will therefore review here only the case $D\ge 4$ and the interested reader can check \cite{GurSch} for the details concerning $D=3$.

The main remark is that in $D\ge 4$, $(D-1)$-dipoles are vertex disjoint. Given any chain (even reduced to one dipole), one can check whether the chain is hooked to the 
left or to the right to another dipole and extend the chain maximally. The maximal chains in a melon free graph are vertex disjoint.

We replace all the maximal chains in a core by chain-vertices. The chain-vertices do not track the length of the chain, but only its type (broken or unbroken)
and its external colors (equal or different). The graphs obtained by this procedure are called \emph{reduced schemes}\footnote{Reduced, as chains of chain
vertices and melonic subgraphs are not allowed}. All the cores leading to the same reduced scheme have the same degree, and we call this common degree the degree of the reduced scheme.
The classification of edge colored graphs is now finished due to the following result.

\begin{theorem}\cite{GurSch}\label{thm:main}
The number of reduced schemes with a fixed degree is finite.
\end{theorem}

The proof of this theorem is a quite technical. It is divided into two parts:
\begin{description}
 \item[Part 1.] {\it The number of chain-vertices in a reduced scheme of degree $\omega$ is bounded.} To prove this statement we iteratively delete the chain-vertices in a scheme (and reconnect the half edges on the left 
 and on the right of the chain-vertex in the unique manner that respects the colorings). There are two cases. First, if the deletion does not separate the graph into connected components
 one can show \cite{GurSch} that the degree strictly decreases. Second, if it separates it, one can show \cite{GurSch} that the degree is distributed between the connected components. However, as a reduced scheme is melon free, 
 both components must have a strictly positive degree. Consequently the number of chain-vertices is bounded linearly in $\omega$.
 
 \item[Part 2.] {\it The number of reduced schemes with degree $\omega$ and bounded number of chain-vertices is finite.} To prove this statement one notices that every chain-vertex admits a minimal realization
 as a shortest chain of $(D-1)$-dipoles. We then map the reduced schemes onto melon free graphs with bounded number of $(D-1)$-dipoles. 
 One can then prove \cite{GurSch} that the number of such graphs is finite by proving they have a finite number of vertices. 
 One starts by bounding the number of faces with two vertices linearly in the number of dipoles and the number of vertices\footnote{In fact one needs to include in the analysis also lower order dipoles, 
 see \cite{GurSch} for details.} by squeezing the maximal number of faces with two vertices without creating further $(D-1)$-dipoles. Together with Eq.~\eqref{eq:truc}, after some work,
 one obtains a linear bound on $k(\cG)$ in terms of the degree $\omega$ and the number of $(D-1)$-dipoles. 
\end{description}

\subsection{The two point function of the SYK--like tensor model at any order in $1/n$}
 
 The classification of edge colored graphs is used in the SYK--like model proposed by Witten as follows. One computes the leading order two point function:
  \[G_{\rm LO}(\tau,\tau') = \sum_{ \cM \in \genfrac{}{}{0pt}{}{\text{two point melonic graphs}}{\text{with external color $0$} } } A^{\cM}(\tau,\tau') \; ,\]
and the leading order broken and unbroken four point functions:
\begin{align*}
& F^{\neq}_{\rm B;LO}(\tau_1,\tau_2,\tau_3,\tau_4) = \sum_{ {\cal U}  \in \genfrac{}{}{0pt}{}{ \text {broken chains with}}{\text{different external colors}}} \hat A^{ {\cal B} ;\neq}(\tau_1,\tau_2,\tau_3,\tau_4)  \;,  \crcr
& F^{ = }_{\rm B;LO}(\tau_1,\tau_2,\tau_3,\tau_4) = \sum_{ {\cal U}  \in \genfrac{}{}{0pt}{}{ \text {broken chains with}}{\text{equal external colors}}} \hat A^{ {\cal B} ;=}(\tau_1,\tau_2,\tau_3,\tau_4) \; , \crcr
& F^{\neq}_{\rm U;LO}(\tau_1,\tau_2,\tau_3,\tau_4) = \sum_{ {\cal U}  \in \genfrac{}{}{0pt}{}{ \text {unbroken chains with}}{\text{different external colors}}} \hat A^{ {\cal U} ;\neq}(\tau_1,\tau_2,\tau_3,\tau_4)  \;, \crcr
& F^{=}_{\rm U;LO}(\tau_1,\tau_2,\tau_3,\tau_4) = \sum_{ {\cal U}  \in \genfrac{}{}{0pt}{}{ \text {unbroken chains with}}{\text{equal external colors}}} \hat A^{ {\cal U} ; = }(\tau_1,\tau_2,\tau_3,\tau_4)  \;,
\end{align*}
(recall that the hat on $\hat A$ signals that the amplitude is evaluated using as propagator the leading order two point function $G_{\rm LO}(\tau,\tau')$).
Of course, in practice it is easier to use self consistency equations implied by the particular combinatorial structure of the melonic graphs and the chains \cite{Polchinski:2016xgd,Jevicki:2016bwu,Maldacena:2016hyu}.

From the connected four point functions one obtains the amputated four point functions:
\[
 \tilde F^{ = (\neq) }_{\rm U(B), LO} (\tau_1,\tau_2,\tau_3,\tau_4) = 
 \int \left( \prod_{i=1}^4 d\tau_i' \right)  \left( \prod_{i=1}^4  (G_{\rm LO} )^{-1} (\tau_i,\tau_i') \right)  F^{ = (\neq) }_{\rm U(B), LO} (\tau_1',\tau_2',\tau_3',\tau_4') \;,
\]
where $ (G_{\rm LO} )^{-1} $ is the operator inverse of the leading order two point function.

Then, Theorem \ref{thm:main} translates as:
 \begin{theorem}\label{thm:mainmic}
  Any order in the $1/n$, the normalized two point function $G(\tau,\tau')$
  in the SYK--like model of \cite{Witten:2016iux} is a {\bf finite} sum of convolutions 
  of a {\bf finite number} of leading order two point functions 
  $ G_{\rm LO}(\tau,\tau')$ and leading order broken and unbroken amputated four point functions:
  \[  \tilde F^{\neq}_{\rm U;LO}(\tau_1,\tau_2,\tau_3,\tau_4) \;, \;\;  \tilde F^{=}_{\rm U;LO}(\tau_1,\tau_2,\tau_3,\tau_4) 
  \;, \;\;  \tilde F^{\neq}_{\rm B;LO}(\tau_1,\tau_2,\tau_3,\tau_4) \;,\;\; \tilde F^{ = }_{\rm B;LO}(\tau_1,\tau_2,\tau_3,\tau_4) \;.\]  
 \end{theorem}

 The crucial point in this theorem is that the number of patterns of convolution at any fixed order in $1/n$  (\emph{i.e.} reduced schemes) is finite. 
 With some effort, one can list all the convolutions allowed at a given order in $1/n$. For instance the first non trivial order 
 convolutions are the lollipop reduced schemes of \cite{GurSch}.
  
The full $1/n$ expansion of arbitrary correlations can be obtained by introducing reduced schemes with several marked edges. The detailed analysis of such reduced schemes
is expected to follow closely the one for reduced schemes with one root edge performed in \cite{GurSch}. While this remains to be done,  
the equivalent of Theorem \ref{thm:mainmic} is expected to hold for arbitrary $2p$ point correlation functions.  

Before concluding, let us comment on the case of a real fermionic field. Much of the combinatorics goes through for non bipartite colored graphs. Of course, the list of allowed convolutions 
(reduced schemes) at the end will change, but this is not very consequential. However, checking rigorously that the number of reduced schemes of fixed degree for non bipartite graphs  
is finite might require some effort. 
 
\section*{Acknowledgements}

The author would like to thank Edward Witten for comments and suggestions.
 
\bibliography{/home/razvan/Desktop/lucru/Ongoing/Refs/Refs.bib}

\end{document}